\documentclass[11pt]{article}
\usepackage{fullpage}
\usepackage{verbatim}
\usepackage{amsmath, amsthm, amssymb}
\usepackage{ifthen}
\usepackage{color, xcolor}

\usepackage{hyperref}
\hypersetup{colorlinks=true, linkcolor=blue, citecolor=blue, anchorcolor=blue, urlcolor=blue}  

\newcommand{\compilehidecomments}{false}

\ifthenelse{ \equal{\compilehidecomments}{true} }{%
	\newcommand{\wei}[1]{}
	\newcommand{\shanghua}[1]{}
	\newcommand{\hanrui}[1]{}
}{
	\newcommand{\wei}[1]{{\color{blue!50!black}  [\text{Wei:} #1]}}
	\newcommand{\shanghua}[1]{{\color{brown!60!black} [\text{Shanghua:} #1]}}
	\newcommand{\hanrui}[1]{{\color{green!60!black} [\text{Hanrui:} #1]}}
}

\newtheorem{theorem}{Theorem}
\newtheorem{lemma}{Lemma}

\newtheorem{definition}{Definition}

\newcommand{\I}{\mathbb{I}}
\newcommand{\R}{\mathbb{R}}

\newcommand{\cB}{\mathcal{B}}

\newcommand{\cD}{\mathcal{D}}

\newcommand{\cG}{\mathcal{G}}
\newcommand{\cH}{\mathcal{H}}

\newcommand{\cT}{\mathcal{T}}

\newcommand{\vg}{\vec{g}}
\newcommand{\vx}{\vec{x}}
\newcommand{\vzero}{\vec{0}}
\newcommand{\vtheta}{\vec{\theta}}

\title{On the Equivalence Between High-Order Network-Influence Frameworks: 
General-Threshold,  Hypergraph-Triggering, and  Logic-Triggering Models}

\author{Wei Chen \\ Microsoft Research \\ Beijing, China \\ {\tt weic@microsoft.com}\and 
        Shang-Hua Teng\footnote{Supported in part by the Simons Foundations' Investigator Award and NSF CCF-1815254. Part of the work was done while the author was visiting 
        	Toyota Technological Institute at Chicago (TTIC).} \\ University of Southern California \\ Los Angeles, CA, USA \\ {\tt shanghua@usc.edu}
        \and
        Hanrui Zhang \\ Duke University \\ Durham, NC, USA \\ {\tt hrzhang@cs.duke.edu}
    }
    
\date{}
\begin{document}
\maketitle

\begin{abstract}
In this paper, we study several 
  high-order network-influence-propagation frameworks and 
  their connection to the classical network diffusion frameworks 
  such as the triggering model and the general threshold model.
In one framework, we use hyperedges to represent 
     {\em many-to-one} influence 
     --- the collective influence of a group of nodes on another node ---
    and define the {\em hypergraph triggering model} as 
    a natural extension to the classical triggering model.
In another framework, we use monotone Boolean functions 
    to capture the diverse logic underlying many-to-one
    influence behaviors, and extend the 
    triggering model to the {\em Boolean-function triggering model}.
We prove that the Boolean-function triggering model, 
   even with {\em refined details of influence logic}, 
   is equivalent to the hypergraph triggering model, and
   both are equivalent to the general threshold model.
Moreover, the general threshold model is {\em optimal} 
   in the number of parameters, among all models with the 
   same expressive power.
We further extend these three equivalent models 
   by introducing correlations among influence propagations 
   on different nodes.
Surprisingly, we discover that 
  while the {\em correlated hypergraph-based model} 
  is still equivalent to the {\em correlated Boolean-function-based model}, 
  the {\em correlated general threshold model} is more 
  restrictive than the two high-order models.
Our study sheds light on high-order network-influence propagations 
  by providing new insight into the group influence behaviors 
  in existing models,
  as well as diverse modeling tools for understanding 
  influence propagations in networks.
\end{abstract}

\section{Introduction}

At the heart of Network Sciences is the problem 
  of {\em network influence}.
This fundamental problem highlights
   the interaction between the dynamics of influence and
   the structures of underlying  networks.
In this problem,
   the collective influence of a group of {\em seed nodes}
   cascades through the network structure, step by step, 
   following a {\em stochastic diffusion process}.
Network influence has been considered in a wide range of fields, 
  including sociology (for studying collective behaviors),
  medical science (for modeling the spread of infectious diseases), 
  political science (for understanding voting behaviors and political influences), 
  legal theory (for analyzing the impact of precedence), 
  material science (for formulating percolation), {\em etc.}
Motivated by the Internet-age application of {\em viral marketing}, 
  Domingos and Richardson
  \cite{domingos01,richardson02} formulated a computational
  optimization problem --- now known as {\em influence maximization} ---
  for modeling the strategic selection of seed nodes aiming (under certain 
  economic constraints) to influence maximum expected adoption of a new 
  product or innovation.
Their work introduced the problem of network influence to computer science.    




\subsection{Operational Network Influence Models}

In their seminal paper \cite{kempe03journal} --- with focus on the algorithmic
  study of influence maximization --- Kempe, Kleinberg, and Tardos
  considered several classes of {\em operational} influence models,
  uniting many formulations proposed previously in 
  multi-disciplinary studies of network influence.
Each model specifies two basic components of network influence: 
\begin{enumerate}
\item {\bf The Underlying Network} -- a directed graph $G = (V,E)$
      which defines for each node $v\in V$, the set $N^-(v) \subseteq V\setminus \{v\}$ 
      of all potential other nodes in the network that may directly influence $v$.
      In the graph, $N^-(v)$ denotes the set of all {\em in-neighbors} of $v$.
\item {\bf The Influence Dynamics} -- a set of parameters and mechanism
      for defining the ``local rules'' of the step-by-step 
      stochastic diffusion process initiated by 
      any selection of seed nodes.
\end{enumerate}
While they use a shared formulation for the underlying network, 
  these models differ in their characterization
  of the stochastic diffusion mechanisms.
The following two models play important roles in studying properties of 
  propagation dynamics and algorithmic
	design.
\begin{itemize}
\item {\bf The General Threshold Model} 
  is perhaps the most natural and individual-based influence model.
For each node $v\in V$, 
(1) in the underlying network $G = (V,E)$, it determines
  the set $N^-(v)$ of other nodes that can directly influence it; 
(2) with its monotone threshold function 
   $f_v: 2^{N^-(v)}\rightarrow [0,1]$,
   it specifies how much each subset of its in-neighbors
   can influence $v$; and 
(3) with its {\em random} threshold valuable, $\theta_v \in [0,1]$,
    it captures how easy it can be influenced.
In the diffusion influence process, 
   $v$ becomes active either immediately when $v$ is in the seed set,
   or at a later time during the process, when
   the $f_v$-value of its active in-neighbors rises above its threshold
   $\theta_v$.

%
\item {\bf The Triggering Model} 
  plays an important role in the algorithm 
  design and analysis for influence maximization \cite{kempe03journal}.
In this model, instead of a threshold function over subsets of its in-neighbors,
  each node $v$ specifies which of its in-neighbors can directly
  influence it according to a distribution over subsets of its in-neighbors.
The randomly chosen subset $T_v$ is called its ``triggering set.''
In the influence diffusion process,
  any member in $T_v$ can influence $v$ (no one in $N^-(v)\setminus T_v$ can influence $v$).
The triggering model can also be equivalently represented as the {\em live-edge graph model} \cite{kempe03journal}.
Analogous to the stochastic infrastructure in percolation theory \cite{BR06},
   the randomly generated sub-graph, $L = \{(u,v): v\in V, u\in T_v\}$,
   defines the edges that influence can cascade through.
\end{itemize}

Different models highlight different features and aspects of 
  network influence that arise in various disciplines.
For example, 
   the general threshold model provides a framework 
   for understanding how properties of local threshold functions
   impact the global influence dynamics, 
   as highlighted by the submodular-threshold conjecture of \cite{kempe03journal},
   which was later settled by \cite{mossel2007}.
The triggering model provides the most direct
   understanding on the underlying submodularity 
   of some network influence settings, because ---  as defined in \cite{kempe03journal}
   --- it is a distribution over deterministic graph-reachability instances.
Different formulations also lead to 
  important sub-models to further focus on various aspects of network influence.
For example, important families of influence model studied in \cite{kempe03journal}
  such as the independent cascade model, 
    the linear threshold model, and
    submodular threshold models
    are naturally characterized within their own respective general frameworks.



\subsection{Our Contributions}

In the real world, 
  direct influences go beyond from one individual to another, 
  and thus, the patterns of interaction are not limited 
  to pair-wise directed graph structures.
Recall that a directed graph over node set $V$ can be viewed as
   a collection of {\em directed edges}. 
In the setting of network influence,
  the underlying directed graph is used to define, for each node $v\in V$,
  the set  $N(v)$ of nodes, each of which potentially 
  has direct influence on  $v$.
In this paper, we study natural high-order extensions of the 
  triggering model.
We will focus on ``network extension'' as well as
   ``stochastic extension.''

First, through the prism of network extension,
   the triggering framework can be applied to 
   Boolean-function-based networks
   and hypergraph-based networks:
\begin{itemize}
\item On the most comprehensive spectrum, direct influences have underlying 
  logical expressions, specifying how each individual can be 
  potentially influence by combinations of others.
This perspective suggests the use of {\em logical networks}.
That is, the state of each node $v$ is represented by a Boolean variable,
  $b_v$, and the network of the potential direct influences
  are represented by each node's collection $B_v$ of 
  Boolean functions over the Boolean variables of other nodes.

\item 
On the relatively simplistic spectrum, 
  direct influences potentially come
  from groups rather than just individuals as specified by direct graphs.
This set-based perspective 
   suggests the use of {\em directed hypergraph} for the 
   underlying interaction network.
A directed hypergraph over $V$ is a collection of 
  {\em directed hyperedges},  each is a pair of 
  a subset and a node, i.e., $(S,v)$, 
  with $S\subseteq V\setminus v$ as the {\em tail set} and $v\in V$ as the head node.
Intuitively, hyperedge $(S,v)$ suggests that when all tail nodes in $S$ are activated, 
	the head node $v$ may be activated due to the collective influence from $S$.
Thus a directed hypergraph can be used to define,
  for each node $v\in V$, the set of groups 
  whose members together can directly influence $v$.
\end{itemize}

In the {\em Boolean-function triggering model} 
   and {\em hypergraph triggering model},
   each node $v$ specifies a distribution, respectively, over 
   subsets of its collection of monotone Boolean functions or 
   incident hyperedges.
Thus, based on the underlying Boolean-function or hypergraph networks.
 the  {\em product} distribution over Boolean functions or hyperedges
 then defines a {\em live} network of Boolean-functions or hyperedges,
 through which influences can cascade.

Our main technical result of this paper is 
  the following model-equivalence theorem,
  which, through the generalization of the classical triggering model 
  to the Boolean-function-triggering or hypergraph-triggering model, 
  characterizes the exact power of high-order network influence.

\begin{theorem}[General Equivalence of Node-Independent Influence Models]
The following three network-influence models 
 ---     the Boolean-function triggering model, 
    the hypergraph triggering model, and 
    the general threshold model
--- 
are equivalent.\footnote{
See Section~\ref{sec:equivalence} for the formal definition of 
  {\em equivalence} between two influence models:
Informally, the model equivalence means that influence propagations
   in the two models produce the same family of distributions 
   of the {\em node-activation sequences} on networks.
}
\end{theorem}

The above theorem illustrates the connection and 
   the difference between the classical triggering model 
  and the general threshold model. 
The triggering model focuses on influence propagation 
  due to pairwise influence relationship between two individuals
	--- each node in a triggering set of a 
   node $v$ is able to influence/activate $v$, and this is
   more restrictive than the general threshold model, 
   which models collective behavior of the neighbors
   through threshold functions on subsets of neighbors of a node.
Mathematically, the classical triggering model is strictly less expressive than
  the general threshold model in part because its
  influence spread is {\em submodular}, while 
  the general threshold model can capture complementarity in influence.

Our equivalence theorem shows that once we include group influence
  --- a basic form of ``complementarity in influence''
     in which a group of neighbors of a node $v$ can collectively
    influence or activate node $v$  --- as modeled by a hyperedge, 
    we immediately obtain the general threshold model
    by extending the triggering model from
    graphs to high-order hypergraph-based networks:
{\em The power of modeling collective behavior through 
   the threshold functions in the general threshold model
   is exactly the same as directly modeling group influence 
  in the hypergraph triggering model.}
Moreover, the equivalence to the Boolean-function triggering model 
  means that any propagation that can be 
  modeled as a probabilistic version of the monotone Boolean-function 
  operation is equivalent to the stochastic group influence modeled
  by the hypergraph triggering model.

We also study the mathematical structures of 
   high-order network influence through the 
   prism of stochastic extension.
Note that all three models in the above 
   theorem are {\em node-independent} models, meaning that 
   the influence to a node $v$ from $v$'s in-neighbors is independent of the influence to other nodes 
	from their in-neighbors.
We can further generalize these models to allow correlations among these influence.
For the Boolean-function triggering model and the hypergraph triggering model, we can naturally extend them
	to their correlated versions such that Boolean-function distributions of different nodes or
	hyperedge subset distributions of different nodes may be correlated.
We call them {\em stochastic Boolean function diffusion (SBFD) model} and 
	{\em stochastic hypergraph diffusion (SHD) model}.
For the general threshold model, we can naturally extended it to the {\em correlated general threshold model} in that
	the threshold distributions of nodes could be correlated.

Now a natural question is that 
  whether the equivalence theorem still holds for these three correlated model
  extensions.
In this paper, we show that the SBFD model is equivalent to the SHD model, 
  but the correlated threshold model is a strict 
  sub-class of the SBFD and SHD models.
This indicates that the general threshold model 
  and its correlated threshold model extension do have
  a unique way of specifying the stochastic propagation behavior, 
  such that when this stochastic behavior
  is not independent across nodes, 
  the resulting behavior is no longer equivalent to other correlated
	diffusion models.
	
In summary, our equivalence theorem connects 
  such group influence models  with the classical 
  general threshold model, and we further understand 
  their connection and difference when we generalize them 
  to correlated diffusion models.
We believe comparative studies of
  diverse high-order network influence frameworks
  --- hypergraph diffusion, Boolean-function diffusion, and 
      general threshold --- 
  not only allow more direct modeling of group influence behaviors 
  but also highlight delicate difference regarding
  the expressiveness of these general network influence models.

\subsection{Related Work}

Kempe et al.~\cite{kempe03journal} are the first summarizing and proposing several classes of influence
	propagation models, which are the foundation for studying the influence maximization and other
	optimization tasks.
In particular, they summarize independent cascade and linear threshold as two basic models from social science
	and statistical physics literature. 
Then they extend both models to the general threshold model and the general cascade model and show these two
	models are equivalent.
They further summarize the triggering model as a generalization of the independent cascade and linear threshold
	model, and show that the triggering model is a strict sub-class of the general threshold model.
They then study the submodularity of these models and propose the use of greedy algorithm on these models
	for the influence maximization task, which is to select $k$ nodes that could maximize the expected number
	of activated nodes, now commonly referred to as the {\em influence spread}.	
For the general threshold model, they conjecture that the influence spread function is submodular if
	every local threshold function is submodular, and this is later proved by Mossel and Roch~\cite{mossel2007}.
	
Since their seminal work, influence diffusion modeling and influence maximization tasks have been extensively studied.
One direction is scalable influence maximization, aiming at designing fast algorithms that can scalable to large
	graphs with millions or even billions of nodes and edges.
Early studies on scalable influence maximization focuses on graph-algorithm-based 
	heuristics~\cite{ChenWY09,WCW12,ChenYZ10,simpath}.
Borgs et al. propose the innovative approach of reverse influence sampling (RIS) that achieves near-linear running time
	and has $1-1/e-\varepsilon$ approximation guarantee for small $\varepsilon > 0$~\cite{BorgsBrautbarChayesLucier}.
RIS approach is subsequently improved by a series of studies~\cite{tang14,tang15,NguyenTD16,TangTXY18} such that
	it can now run on billion-edge graphs in just seconds to find 50 seed nodes. 
It turns out that so far the model that allows the RIS approach is exactly the triggering model.

There are many other directions in influence propagation and influence maximization,
	for example, competitive and complementary influence propagation~\cite{BAA11,HeSCJ12,lu2015competition}, 
	seed set minimization~\cite{amit_www11,long2011minimizing,ZhangCSWZ14}, 
	profit maximization~\cite{LL12,TTY16}, online influence maximization~\cite{CWYW16,LeiMMCS15,Wen2016,WC17}, etc.
In some of these studies, influence propagation models are extended and submodularity of the models are studied.
But most of the studies are built on the triggering model or the general threshold model, or one of the more
	specific models such as the independent cascade model or the linear threshold model.
We refer to the general survey work in this area for more detailed coverage on these and other related topics
	\cite{chen2013information,LiFWT18}.
	
Zhu et al. \cite{ZZGWY18} has proposed an influence diffusion model incorporating hyperedges. 
In particular, the hyperedge is the same as modeled in this paper, which contains multiple tail nodes and
	one head node, with the same interpretation that all tail nodes together could activate the head node.
In their model, each edge or hyperedge has an independent probability to be live, and thus it corresponds to
	the generalization of the independent cascade model to directed hypergraphs.
This means their model is a special case of the hypergraph triggering model defined in this paper.
They study the influence maximization problem under the hypergraph independent cascade model.
Since the model is neither submodular nor supermodular, they use sandwich approximation
	technique~\cite{lu2015competition} to find submodular upper and lower bounds of the original influence
	spread function, and then apply RIS approach to find the seed set.

Hypergraphs are used in the multimedia recommendation context~\cite{BuTCWWZH10,SperliAMP16}, where
	a heterogeneous network contains different types of nodes such as users, media contents, and tags, 
	and hyperedges represent relationship among these nodes such as friend relationship among users, 
	similarity relationship among media contents, and tagging relationship among user, content and tags.
Amato et al. \cite{AMPS17} apply influence maximization algorithms to such multimedia hypergraphs, but
	they are not really modeling influence propagation directly on the hypergraph.
Instead, they simply transform the hypergraph into a bipartite regular graph and then apply influence maximization
	algorithms on the regular graph.

\section{Preliminaries and Model Equivalence of Network Influence}


In this section, we review two concrete
  network influence models: the general threshold model 
  and triggering model.
We then define an {\em abstract} stochastic-diffusion framework
  as the mathematical basis for formalizing
  {\em equivalence} among concrete network-influence models.

\subsection{General Threshold Model and Triggering Model}

The general threshold model and triggering model are two 
  concrete network-influence models originally formulated and 
  studied by Kempe et al.\ in~\cite{kempe03journal}.
The general threshold model is built on 
  nodal set functions, which assign each subset of a ground
  set a (non-negative) real value.
Recall that for a set function $f: 2^V \rightarrow \R$, we say that 
	(a) $f$ is {\em monotone} if
	for all $S\subseteq T \subseteq V$, $f(S) \le f(T)$; 
	(b) $f$ is {\em normalized} if $f(\emptyset) = 0$; and
   	(c) $f$ is {\em submodular} if 
    for all $S\subseteq T \subseteq V$ and all $v\in V\setminus T$,
		$f(S\cup \{v\}) - f(S) \ge f(T\cup \{v\}) - f(T)$.

\begin{definition}[General Threshold Model]
\label{def:GT}
	A {\em general threshold model} 
	is defined by a tuple $(V, \{f_v \}_{v\in V})$, 
  where $V$ is the node set
		and 
	$f_v: 2^{V\setminus \{v\}} \rightarrow [0,1]$ is the threshold function for $v \in V$, and
	all $f_v$'s are monotone and normalized.
	The diffusion process proceeds as follows.
	At the beginning, every node $v$ samples a threshold $\theta_v$ uniformly at random
	from $[0,1]$ (denoted $\theta_v \sim U[0,1]$).
	Given a seed set $S_0 \subseteq V$, at time step $t=0$, all nodes in $S_0$ are activated 
	and all nodes not in $S_0$ are inactive.
	Let $S_t$ be the set of active nodes by time $t$.
	At any time step $t>0$, if $v\in S_{t-1}$, then $v\in S_t$;
	if $v\not\in S_{t-1}$, then $v$ is activated at time $t$ (i.e. $v\in S_t \setminus S_{t-1}$)
	if $f_v(S_{t-1}) \ge \theta_v$.
\end{definition}

The triggering model is a subclass of 
  the general threshold model, and it contains the 
  well-studied independent cascade and 
  linear threshold models as special cases.
	
\begin{definition}[Triggering Model]
\label{def:triggering}
A {\em triggering model} is defined by a tuple 
  $(V, \{\cT_v \}_{v\in V})$, where $V$ is the node set
  and $\cT_v$ is a triggering distribution 
  over the subsets of $V\setminus \{v\}$.
	The diffusion process proceeds as follows.
At the beginning, every node $v$ samples a triggering set $T_v \sim \cT_v$.
Given a seed set $S_0$, at time step $t=0$, all nodes in $S_0$ are activated 
and all nodes not in $S_0$ are inactive.
Let $S_i$ be the set of active nodes by time $t$.
At any time step $t>0$, if $v\in S_{t-1}$, then $v\in S_t$;
if $v\not\in S_{t-1}$, then $v$ is activated at time $t$ (i.e. $v\in S_t \setminus S_{t-1}$)
if $S_{t-1} \cap T_v \ne \emptyset$.
\end{definition}

Note that the general threshold model and 
  triggering model can also be defined on a graph
 $G=(V,E)$ with the predefined directed edge set $E$.
In this case, the threshold functions $f_v$ is 
  defined on the subsets of in-neighbors of $v$, and
  the triggering set distribution $\cT_v$ is 
  a distribution on the subsets of in-neighbors of $v$.
However, such $f_v$'s and $\cT_v$'s can also be 
  extended to $V\setminus \{v\}$, and thus
  we do not need to explicitly refer to the edge set $E$.

\subsection{Model Equivalence of Network Influence: 
  Abstract Stochastic Diffusion} \label{sec:equivalence}

Note that concrete network-influence models ---  such as
   the triggering model and the general threshold model ---
   not only define
   the probabilistic time-series of node sets $S_1, S_2, \ldots$ 
   as the result of influence propagation/activation 
    from each starting seed set $S_0 \subseteq V$,
   but also provide details of the underlying diffusion process concerning
   how activated nodes influence other nodes in each time step.
To comparatively study different influence models --- particularly
  regarding their equivalence  --- 
  we consider an abstract diffusion framework
  which only focuses on the {\em stochastic profile}
  of activation sequences of influence processes:
We say that a sequence of sets 
  $S_0, S_1, S_2, \ldots$ ($S_t \subseteq V$ for all $t\ge 0$) 
is {\em progressive} if 
\begin{itemize}
\item [(a)] for all $t\ge 0$, $S_t \subseteq S_{t+1}$;  
\item [(b)] for any $t\ge 0$, if $S_t = S_{t+1}$, 
   then for all $t' > t$, $S_{t'} = S_t$; and
\item [(c)] $S_0 \ne \emptyset$.
\end{itemize}
In this paper, we only study influence diffusion that generates such progressive sequences, 
which means that --- like in \cite{kempe03journal} ---  (a) once a node is influenced, it remains influenced forever, (b)
influence propagation should occur at every step --- if in one step there is no newly
influenced nodes, then the influence propagation stops, and (c) influence has to be
started from some nonempty seed set $S_0$ and cannot be generated spontaneously.
Let $n=|V|$. 
For such progressive sequences, it is clear that the influence propagation stops within
at most $n-1$ steps, and thus we only need a sequence $S_0, S_1, \ldots, S_{n-1}$ to 
represent it.
The formal definition of the (abstract) 
  stochastic diffusion model is given below
	(a similar definition is originally given 
     in~\cite{chen2013information}).

\begin{definition}[Abstract Stochastic Diffusion]
	\label{def:SDM}
	Given a node set $V$ of size $n$, a {\em stochastic diffusion model}  $\cD_V$ on node
	set $V$ is
	a mapping from every nonempty seed set $S_0 \in 2^V \setminus \{\emptyset \}$ to
	a distribution $\cD_{V, S_0}$ over all progressive sequences 
	$(S_0, S_1, \ldots, S_{n-1})$ starting from $S_0$.
	That is, for every nonempty seed set $S_0$, $\cD_{V, S_0}(S_0, S_1, \ldots, S_{n-1})$ 
	gives the probability
	that a progressive sequence $(S_0, S_1, \ldots, S_{n-1})$ is generated from the diffusion
	process starting from seed set $S_0$.
\end{definition}

Comparing to concrete influence models, such as the triggering and 
  general threshold models,
  the above abstract stochastic diffusion model 
  focuses on the result of influence processes at each steps
  rather than the detailed logic or mechanism underlying the influence processes.
It distills the influence processes into 
  the distribution $\cD_{V, S_0}$ over progressive sequences 
  for every seed set $S_0$,
  and provides a basis for 
  reasoning about the equivalence of concrete diffusion models
  and network-influence frameworks:

\begin{definition}[Equivalence of Stochastic Diffusion Models]
	\label{def:equivSDM}
Two stochastic network-influence models over the same node set $V$
  are {\em equivalent} if in their induced abstract stochastic
  diffusion models $\cD_V$ and $\cD'_V$, 
  the distributions over progressive sequences on every seed set are the same,
  i.e., for all $S_0 \in 2^V \setminus \{\emptyset \}$, 
  $\cD_{V,S_0} = \cD'_{V,S_0}$.
\end{definition}
%

\section{Stochastic Hypergraph Diffusion Model}

We now introduce one of the main subjects of the study in this paper, the
	diffusion model based on directed hypergraphs.
Given a set of nodes $V$, a {\em directed hyperedge} $h = (U, v)$ represents
	a connection from a subset $U \subseteq V \setminus \{v\}$ to a node $v$.
We call $U$ as the {\em tail set}
   of hyperedge $h$ (or simply tails), and $v$ as 
   the {\em head} of hyperedge $h$.
A directed {\em hypergraph} is represented by $G=(V,H)$, 
  where $V$ is a set of $n$ vertices
  or nodes, and $H$ is a set of directed hyperedges on $V$.
Note that if the tail sets of all hyperedges are singletons, 
  the hyperedge graph degenerates to a normal directed graph.
	
In the context of information or influence propagation, 
  a hyperedge $h = (U,v)$ represents a propagation step: 
  if {\em all} nodes in $U$ are activated in the previous step, then $v$
	is activated in the current step.
Formally, given a hypergraph $G$ and a nonempty 
  seed set $S_0\subseteq V$, the deterministic 
	propagation on $G$ from $S_0$ carries out as follows.
\begin{quote}
At time $t=0$, all nodes in $S_0$ are activated and all nodes in $V\setminus S_0$ are inactive.
Let $S_t$ denote the set of nodes that have been activated by time $t$.
At any time step $t > 0$, if $v\in S_{t-1}$, then $v \in S_t$; 
	if $v \not\in S_{t-1}$, and there exists a hyperedge $h = (U,v) \in H$ such that
	$U \subseteq S_{t-1}$, then $v \in S_t$; for all other $v$ not belonging to the above 
	two cases, $v \not\in S_t$.
\end{quote}
Thus, given $G$ and $S_0$, 
  we generate a sequence of node sets $S_0, S_1, S_2, \ldots$, 
	and it is easy to verify that this set sequence is progressive.
Note that if $G$ degenerates to a normal directed graph, 
  then sequence $S_0, S_1, \ldots, S_{n-1}$
	corresponds to the normal breadth-first-search (BFS) sequence starting from node set $S_0$:
	$S_t$ is the set of nodes reached from $S_0$ in at most $t$ steps.
For simplicity, we also refer to the set sequence $S_0, S_1, \ldots, S_{n-1}$ generated from
	the hypergraph $G$ and seed set $S_0$ as the {\em BFS sequence} of $G$ and $S_0$,
	and we denote $S_t = \Gamma_t(G, S_0)$ as the set of nodes $S_0$ could reach in hypergraph
	$G$ within $t$ steps.
	
The stochastic diffusion under the hypergraph model is 
  by a probabilistic sampling of 
  hypergraphs followed by the deterministic 
  propagation on the sampled hypergraph.
The formal definition of the model is given below:

\begin{definition}[Stochastic Hypergraph Diffusion (SHD) Model]
\label{def:SHD}
A {\em stochastic hypergraph diffusion (SHD)} model 
  is defined as a tuple $(V, \cG_V)$, where
  $V$ is a set of $n$ nodes, and 
  $\cG_V$ is a distribution of hypergraphs on $V$.
Given a seed set $S_0 \subseteq V$, 
  we first sample a hypergraph $G\sim \cG_V$, and then
  the propagation from $S_0$ proceeds on the hypergraph $G$ to generate the
  BFS sequence $S_0, S_1, \ldots, S_{n-1}$, 
  where $S_t = \Gamma_t(G, S_0)$ for all $t=1, 2, \ldots, n-1$.
\end{definition}

It is clear that the SHD model 
    induces an abstract stochastic diffusion model 
    as defined in Definition~\ref{def:SDM}.
Note that the SHD model given in Definition \ref{def:SHD} 
  is very general, allowing arbitrary distributions over hypergraphs. 
In this paper, we further investigate
  an important subclass of the SHD model 
  that has the following node-independence property.
We say that a hypergraph distribution $\cG_V$ is {\em node-independent} if
  $\cG_V = \times_{v\in V} \cH_v$, 
  where $\cH_v$ is a distribution over subsets of hyperedges 
	with head $v$.
In other words, a sample $G\sim \cG_v$ can 
  be obtained by independently sampling subsets of hyperedges
  pointing to $v$ according to $\cH_v$, for all $v\in V$, 
  and then combining all 
	these hyperedges together to form the hypergraph.
When we restrict that $\cH_v$ only has supports on 
  subsets of normal edges, not hyperedges, 
  the node-independent SHD model degenerates to the triggering model
	given in Definition~\ref{def:triggering}.
Therefore, we refer this hypergraph version as the 
  {\em hypergraph triggering model}.

\begin{definition}[Hypergraph Triggering Model]
A {\em hypergraph triggering model} is a stochastic hypergraph diffusion model
	$(V, \cG_V)$ where $\cG_V$ is node-independent, i.e. $\cG_v = \times_{v\in V} \cH_v$,
	where $\cH_v$ is a distribution over subsets of hyperedges 
	with head $v$.
\end{definition}

\section{Stochastic Boolean-Function Diffusion Model 
   and Its Equivalence to the SHD Model}

%
%

To unify the study of the stochastic hypergraph diffusion model 
  and the general threshold
  model, in this section, we further consider a general class of
	diffusion models that we will refer to as 
        the {\em stochastic Boolean-function diffusion} (SBFD) model.
Intuitively, we want to study a large class of 
  stochastic diffusion models in which
 the transition from $S_{t-1}$ to $S_t$ in each 
  propagation step $t \ge 1$ is governed by
  the same transition function, and 
  this transition function is probabilistically selected
	from a distribution before the propagation starts.
Both the SHD model and the general threshold model 
   belongs to this class: In the SHD model,
	the transition function is given by the hypergraph 
        and the distribution is on the hypergraph, 
        while in the general threshold model,
	the transition is determined by the threshold functions 
        $\{f_v\}_{v\in V}$ and the random thresholds $\{\theta_v\}_{v\in V}$
         drawn from the uniform distribution on the thresholds.
While the transition functions
    can be viewed as maps from subsets of $V$ to other subsets of $V$, 
    they can also be defined by Boolean functions.
This is because the state of whether a node
	is active or not can be represented as a Boolean variable.
Thus, it is natural
  to use Boolean functions to represent the logic 
  underlying influence propagation.
We consider the Boolean-function representation also because 
  --- as it will be clear later ---  it is easier to connect 
  with the hypergraph representation.

\subsection{Stochastic Boolean Function Diffusion Model}

For a Boolean vector $\vx \in \{0,1\}^V$, we use $x_v$ to represents its value in the
	dimension corresponding to $v$.
For $\vx$, we use $S^{\vx}$ to represent its corresponding
set, i.e. $S^{\vx} = \{v\in V\mid x_v = 1\}$.
Conversely, for a subset $S \subseteq V$, we use $\vx^S$ to represent the corresponding
Boolean vector, i.e. $x^S_v = 1$ if $v\in S$, and $x^S_v = 0$ if $v\not\in S$.
We use $\vx_{-v}$ to represent the projection of vector $\vx$ to $V\setminus \{v\}$, i.e.
removing the dimension corresponding to $v$.

For each $v\in V$, we use Boolean variable $x_v \in \{0,1\}$ to represent $v$'s state:
	$x_v=0$ means $v$ is inactive and $x_v=1$ means $v$ is activated.
Then for each $v \in V$, we associate $v$ with a {\em Boolean activation function}
	$g_v : \{0,1\}^{V\setminus \{v\}} \rightarrow \{0,1\}$ to represent
	how $v$ is influenced by other nodes in the network.
We require that Boolean function $g_v$ be monotone, meaning that for two vectors
	$\vx, \vx' \in \{0,1\}^{V\setminus \{v\}}$, if $x_u \le x'_u$ for all $u\in V\setminus \{v\}$,
	then $g_v(\vx) \le g_v(\vx')$.
Moreover, we require that $g_v$ is normalized, that is, $g_v(\vzero) = 0$.

The diffusion process in $V$ can be represented as a sequence of Boolean vectors
	$\vx_0, \vx_1, \ldots, \vx_{n-1}$, where $S^{\vx_t}$ corresponds to the set of active nodes
	by time $t$.
The transition from $\vx_t$ to $\vx_{t+1}$ is governed by the following transition function.

\begin{definition}[Boolean Transition Function]
\label{def:btf}
Given a node set $V$ and Boolean activation functions $\{g_v\}_{v\in V}$, we define the
	{\em Boolean transition function} $\vg : \{0,1\}^V \rightarrow \{0,1\}^V $ as follows:
	for every $v\in V$, for every $\vx \in \{0,1\}^V$, the dimension of $\vg(\vx)$ corresponding
	to $v$ is defined as
	$\vg_v(\vx) = x_v \vee g_v(\vx_{-v})$. 
\end{definition}
In the above definition, $\vg_v(\vx) = x_v \vee g_v(\vx_{-v})$ means that
	if $v$ is already activated, then $v$ stays active, and if not,
	then we look at the states of other nodes $\vx_{-v}$, and if they could activate $v$
	according to the Boolean activation function $g_v$, then $v$ is activated.
Given the transition function $\vg$, the diffusion is simply by repeatedly applying $\vg$, 
	that is, $\vx_1 = \vg(\vx_0), \vx_2 = \vg(\vx_1), \ldots, \vx_{n-1} = \vg(\vx_{n-2})$.
Thus, the Boolean transition function $\vg$ is the Boolean representation of the transition
	function we discussed at the beginning of this section.
The stochastic version is then by randomly sampling Boolean activation functions, 
	as defined below.
	
\begin{definition}[Stochastic Boolean Function Diffusion Model]
\label{def:SBFD}
A {\em stochastic Boolean function diffusion (SBFD)} model is defined as a tuple $(V, \cB_V)$, where
$V$ is a set of $n$ nodes, and $\cB_V$ is a distribution on the set of Boolean activation
	functions $\{g_v\}_{v\in V}$.
Given a seed set $S_0 \subseteq V$, we first sample a set of 
	Boolean activation functions $\{g_v\}_{v\in V} \sim \cB_V$, 
	and then the propagation from $S_0$ proceeds 
	by repeatedly applying Boolean transition function $\vg$ on $\vx^{S_0}$, 
	i.e. $\vx_{t+1} = \vg(\vx_t)$ for all $t\ge 0$, 
	where
	$\vg$ is defined in Definition~\ref{def:btf}.
	The equivalent set sequence is $S_t = S^{\vx_t}$.
\end{definition}

Like in SHD, we say that an SBFD model $(V, \cB_V)$ is 
   node-independent if 
   $\cB_V = \times_{v\in V} \cB_v$, 
   where $\cB_v$ is a distribution on Boolean
   activation functions $g_v$'s for node $v$.
To be consistent with the terminology in the 
  hypergraph model, we also call node-independent
  SBFD model as {\em Boolean-function triggering model}.

Illustrating that the SBFD model is indeed a 
   general model encompassing both
   the SHD model and the general threshold model,
 we first show below that each SHD model and the general threshold model
	can be directly expressed by an SBFD model.

\begin{lemma}[Expressing SHD by SBFD]
	\label{lem:SHDisSBFD}
	Any SHD model can be represented as an SBFD model.
\end{lemma}
\begin{proof}
	Consider an SHD model $(V, \cG_V)$.
	Let $G=(V,H)$ be a sample hypergraph according to distribution $\cG_V$.
	Given a seed set $S_0$, notice that the propagation on $G$ from $S_0$ can be equivalently
	described as repeatedly applying function $\Gamma_1(G,\cdot)$, i.e.
	$S_t = \Gamma_1(G,S_{t-1})$ for all $t\ge 0$.
	Then we just need to write transition function $\Gamma_1$ in the form of Boolean activation
		functions $\{g_v\}_{v\in V}$.
	In fact, for every $v\in V$, let $(U_1, v), (U_2, v), \ldots, (U_k,v)$ 
		be the set of hyperedges in $H$ that point to $v$.
	Then we can define the Boolean activation function 
	$g_v(\vx) = \bigvee_{j=1}^k \bigwedge_{u\in U_j} x_u$.
	Obviously $g_v$ is monotone and normalized.
	Let $\vg$ be the Boolean transition function derived from $\{g_v\}_{v\in V}$ according
	to Definition~\ref{def:btf}.
	For every $S \in 2^V \setminus \{\emptyset \}$, let $T = \Gamma_1(G,S)$, i.e. 
	$T$ is the set of nodes activated in $G$ in one step from source set $S$.
	It is easy to check that $\vx^T = \vg(\vx^S)$.
	Thus the propagation on hypergraph $G$ is exactly the same as the propagation 
	following the Boolean activation functions $\{g_v \}_{v\in V}$.
	Thus, we only need to set the Boolean function distribution $\cB_V$ to be corresponding to
	$\cG_V$, then the SHD model $(V, \cG_V)$ becomes the SBFD model $(V,\cB_V)$.
\end{proof}

\begin{lemma}[Expressing General Threshold Model by SBFD Model]
\label{lem:GTisSBFD}
Any general threshold model can be represented as a node-independent SBFD model.
\end{lemma}
\begin{proof}
Consider a general threshold model $(V, \{f_v \}_{v\in V})$.
Let $\vtheta=(\theta_v)_{v\in V}$ be a sampled threshold vector, i.e. $\theta_v\sim U[0,1]$.
For each $v\in V$, we fix its threshold $\theta_v \in [0, 1]$.
Since $\theta_v=0$ has zero-measure, without loss of generality we consider $\theta_v > 0$.

For any $S \in 2^V \setminus \{\emptyset \}$ and any $v\in V$, $v$ would be activated by $S$ in one step
	if either $v\in S$ or $v\not\in S$ and $f_v(S) \ge \theta_v$.
This can be exactly represented as the Boolean activation function 
	$g^{\theta_v}_v(\vx) = \I\{f_v(S^{\vx}) \ge \theta_v \}$, where $\I$ is the indicator function, and
	we use superscript $\theta_v$ to explicitly denote that this Boolean activation function is
	determined by threshold $\theta_v$.
When $\theta_v > 0$, we have $g^{\theta_v}_v$ as a monotone and normalized function, since
	$f_v$ is monotone and normalized.
Let $\vg^{\,\vtheta}$ be the Boolean transition function derived from $\{g^{\theta_v}_v \}_{v\in V}$
	as in Definition~\ref{def:btf}.	
It is easy to check that, with the fixed $\vtheta$, for any seed set $S_0$, the propagation sequence
	generated by the general threshold model is equivalently generated by repeated applying 
	$\vg^{\,\vtheta}$, i.e. $\vx_0 = \vx^{S_0}$, $\vx_t = \vg^{\,\vtheta}(\vx_{t-1})$,
	$S_t = S^{\vx_t}$, for all $t\ge 1$.
Therefore, the general threshold model under the fixed threshold vector $\vtheta$ is the same
	as the SBFD model under the fixed Boolean activation functions $\{g^{\theta_v}_v\}_{v\in V}$.
Finally, we define $\cB_v$ as the distribution of $g^{\theta_v}_v$ when $\theta_v\sim U[0,1]$.
Since $\theta_v$'s are mutually independent, 
	this would give us a node-independent SBFD model $(V, \times_{v\in V} \cB_v )$ that exactly
	represents the corresponding general threshold model $(V, \{f_v \}_{v\in V})$.
We remark that for a given $V$, there are only a finite number of Boolean activation functions,
	and thus for different $\theta_v$, we may have the same function
	$g^{\theta_v}_v$, then this $g^{\theta_v}_v$ will have a probability mass 
	in $\cB_v$. 
The actual probability mass value can be determined from function $f_v$, but we omit the
	discussion here.
\end{proof}

\subsection{Equivalence of the SHD model and the SBFD model}

The following theorem shows the SHD model is actually equivalent to the SBFD model, meaning
	that the two have exactly the same expressive power.

\begin{theorem}[Equivalence between SHD and SBFD Models]
	\label{thm:SHDvsSBFD}
Every hypergraph $G=(V,H)$ has a corresponding set of Boolean activation functions $\{g_v\}_{v\in V}$
	such that they generate the same progressive sequence for every nonempty seed set $S_0$, and
	vice versa.
A direct consequence is that every stochastic hypergraph diffusion model $(V, \cG_V)$ has a corresponding equivalent
	stochastic Boolean function diffusion model $(V, \cB_V)$, and vice versa.
\end{theorem}
\begin{proof}
%
For the direction from a hypergraph $G=(V,H)$ to a set of Boolean activation functions
	$\{g_v\}_{v\in V}$, it has been shown in the proof of Lemma~\ref{lem:SHDisSBFD}.
%

Conversely, suppose that we have Boolean activation function $g_v$ for every $v\in V$.
It is well known (see, e.g., \cite{gurvich1999generating}) that the reduced disjunctive normal form of any nontrivial (i.e., not constantly $0$ or $1$)
    monotone Boolean function does not contain negations of variables.
Thus, there exists $U_1, U_2, \ldots, U_k$ such that 
	$g_v(\vx) = \bigvee_{j=1}^k \bigwedge_{u\in U_j} x_u$.
Then we can construct $(U_1, v), (U_2, v), \ldots, (U_k,v)$ as the set of hyperedges
	pointing to $v$.
Again they would generate the same progressive sequence for every seed set $S_0$.
Thus, the theorem holds.
\end{proof}

%
%
%

\section{Equivalence between Hypergraph-Triggering Model and 
  General Threshold Model}

Since stochastic hypergraph diffusion
     model is equivalent to the stochastic Boolean-function diffusion
     model, for convenience, our analysis below is 
     based on the SBFD model.
Recall that Lemma~\ref{lem:GTisSBFD} 
   shows that any general threshold model 
   can be expressed by equivalent node-independent SBFD model.
To show the reverse direction, we first prove the 
   following useful result, which shows that
	under the node-independent 
	SBFD model, the progressive sequence 
        distribution is fully determined by
	the set-node activation probabilities.
In the following lemmas, we set $S_{-1} = \emptyset$ for convenience.

\begin{lemma}
\label{lem:setvertexSBFD}
The distribution of progressive sequences in a node-independent
	SBFD model $(V, \times_{v\in V} \cB_v)$
	is fully determined by set-node activation probabilities  
	$\{\Pr\{g_v(\vx^S) = 1\}\}_{S, v}$ for all $S\subseteq V$ and $v\in V$.
More specifically, we have
\begin{align}
	& \Pr_{g_v \sim \cB_v, v\in V}\{\mbox{$S_0,S_1, \ldots, S_{n-1}$ is generated}  \} 
		\nonumber \\
	=\ &\prod_{t=1}^{n-1} \prod_{v \in S_t \setminus S_{t-1}} 
\left(	\Pr_{g_v \sim \cB_v} \left\{
g_v(\vx^{S_{t-1}}_{-v}) = 1 \right\} - 
\Pr_{g_v \sim \cB_v} \left\{ g_v(\vx^{S_{t-2}}_{-v}) = 1 \right\} \right) \cdot
\prod_{v \notin S_{n-1}}
\left(1- \Pr_{g_v \sim \cB_v}\left\{g_v(\vx^{S_{n-2}}_{-v}) = 1 \right\} \right).
\label{eq:fullydetermined}
\end{align}

\end{lemma}
\begin{proof}
	Fix a propagation sequence $(S_0, \dots, S_{n-1})$.
	Let $\vx_t = \vx^{S_t}$, for $t=-1, 0, 1, \ldots, n-1$.
	Recall from Definition~\ref{def:btf} that $\vg$ is the Boolean transition function
	defined from the model $(V, \times_{v\in V} \cB_v)$, with $\vg_v(\vx) = x_v \vee g_v(\vx_{-v})$.
	\begin{align*}
	& \Pr_{g_v \sim \cB_v, v\in V}\{\mbox{$S_0,S_1, \ldots, S_{n-1}$ is generated}  \} \\
	=\ & \Pr_{g_v \sim \cB_v, v\in V}\left\{\bigwedge_{t=1}^{n-1} \vx_t = \vg(\vx_{t-1})\right\} \\
	=\ & \Pr_{g_v \sim \cB_v, v\in V}
	\left\{\bigwedge_{t=1}^{n-1} \left(
		\left(\bigwedge_{v \in S_t \setminus S_{t-1}} g_v(\vx_{t-1.-v}) = 1\right) \wedge 
		\left(\bigwedge_{v \notin S_t} g_v(\vx_{t-1.-v}) = 0 \right) \right) \right\} \\
	=\ & \Pr_{g_v \sim \cB_v, v\in V}
	\left\{\bigwedge_{t=1}^{n-1} 
		\left(\bigwedge_{v \in S_t \setminus S_{t-1}} 
			(g_v(\vx_{t-1.-v}) = 1 \wedge g_v(\vx_{t-2.-v}) = 0)
			\right) \wedge 
		\left(\bigwedge_{v \notin S_{n-1}} g_v(\vx_{n-2.-v}) = 0 \right)  \right\} \\
	=\ &\prod_{t=1}^{n-1} \prod_{v \in S_t \setminus S_{t-1}} 
		\Pr_{g_v \sim \cB_v} \left\{
			g_v(\vx^{S_{t-1}}_{-v}) = 1 \wedge g_v(\vx^{S_{t-2}}_{-v})) = 0 \right\}  \cdot
	 \prod_{v \notin S_{n-1}} \Pr_{g_v \sim \cB_v}\left\{g_v(\vx^{S_{n-2}}_{-v}) = 0 \right\} \\
	=\ &\prod_{t=1}^{n-1} \prod_{v \in S_t \setminus S_{t-1}} 
		\left(	\Pr_{g_v \sim \cB_v} \left\{
		g_v(\vx^{S_{t-1}}_{-v}) = 1 \right\} - 
		\Pr_{g_v \sim \cB_v} \left\{ g_v(\vx^{S_{t-2}}_{-v})) = 1 \right\} \right) \cdot
		\prod_{v \notin S_{n-1}}
			\left(1- \Pr_{g_v \sim \cB_v}\left\{g_v(\vx^{S_{n-2}}_{-v}) = 1 \right\} \right).
	\end{align*}
From the above derivation, it is clear that the probability of generating the sequence
	$S_0,S_1, \ldots, S_{n-1}$ is fully determined by 
	the set-node activation probabilities $\Pr\{g_v(\vx^S) = 1\}$'s for all 
	$S\subseteq V$ and all $v\in V$.
\end{proof}

We are now ready to show the following theorem 
	on the equivalence between the general threshold model and
	the hypergraph triggering model.
	
\begin{theorem}[Equivalence between General Threshold Model and Hypergraph Triggering Model]
	\label{thm:GTvsHT}
Every general threshold model $(V, \{f_v\}_{v\in V})$ has a corresponding equivalent
	hypergraph triggering model $(V, \cG_V)$, and vice versa.
\end{theorem}
\begin{proof}
As we have discussed before, we will use the node-independent SBFD model 
  to show the equivalence.
Lemma~\ref{lem:GTisSBFD} already shows that any general threshold model has an
equivalent representation as a node-independent SBFD model.
Thus, we only need to show the reverse direction, that is, 
	every node-independent SBFD model $(V, \cB_V) = (V,\times_{v\in V} \cB_v)$ has a 
	corresponding equivalent general threshold model $(V, \{f_v\}_{v\in V})$.
Given $\times_{v\in V} \cB_v$, we construct the 
	threshold functions $\{f_v\}_{v\in V}$ such that
	\begin{equation} \label{eq:equivactivateprob}
	f_v(S) = \Pr_{g_v \sim \cB_v}\{g_v(\vx^S) = 1\}, \forall S \subseteq V.
	\end{equation}
By Lemma~\ref{lem:setvertexSBFD}, the distribution of progressive
	sequences in the SBFD model $(V,\times_{v\in V} \cB_v)$ is fully determined by 
	$\Pr_{g_v \sim \cB_v}\{g_v(\vx^S) = 1\}$ as in Eq.~\eqref{eq:fullydetermined}.
By Lemma~\ref{lem:GTisSBFD} and its proof, we see that
	the general threshold model $(V, \{f_v\}_{v\in V})$ can
	be represented as a SBFD model $(V,\times_{v\in V} \cB'_v)$, where
	$\cB'_v$ is the distribution on $g^{\theta_v}_v$ with $\theta_v \sim U[0,1]$.
Note that $f_v(S)$ is exactly the set-node activation probability in 
	the model $(V,\times_{v\in V} \cB'_v)$, since
\begin{align} \label{eq:equivactivateprob2}
& \Pr_{g^{\theta_v}_v \sim \cB'_v}\{g^{\theta_v}_v(\vx^S) = 1\} 
 = \Pr_{\theta_v \sim U[0,1]}\{f_v(S) \ge \theta_v\} = f_v(S).
\end{align}
Then by comparing Eq.~\eqref{eq:equivactivateprob} with Eq.\eqref{eq:equivactivateprob2} and applying Lemma~\ref{lem:setvertexSBFD} again, we know that the distribution of
	progressive sequences generated by the $(V,\times_{v\in V} \cB_v)$ is the same as 
	the distribution of progressive sequences generated by  $(V,\times_{v\in V} \cB'_v)$,
	the latter of which is the same as the distribution of progressive sequences generated by 
	the general threshold model $(V,\{f_v\}_{v\in V})$.
Hence, the two models are equivalent according to Definition~\ref{def:equivSDM}.
\end{proof}

A follow-up of the above equivalence result is that the general threshold model is the one
	using the minimum number of parameters, as shown below.

\begin{theorem}
General threshold model is the one with the minimum number of parameters among all 
	equivalent hypergraph triggering models.
\end{theorem}
\begin{proof}
By Theorem~\ref{thm:SHDvsSBFD} and Lemma~\ref{lem:setvertexSBFD}, any hypergraph triggering model
	(or equivalent node-independent SBFD model) is fully determined by
	the set-node activation probabilities $\Pr_{g_v \sim \cB_v}\{g_v(\vx^S) = 1\}$
	for all $S\in 2^V \setminus \{\emptyset \}$ and $v\in V$.
Moreover, for every such $S$ and $v$, we could choose a different value for
	$\Pr_{g_v \sim \cB_v}\{g_v(\vx^S) = 1\}$ (while satisfying that $g_v$ is
	normalized and monotone), and thus we need at least this number of parameters to determine
	these set-node activation probabilities.
For the general threshold model, $f_v(S) = \Pr_{\theta_v\sim U[0,1]} \{ f_v(S) \ge \theta_v \}$
	provide exactly these parameters.
Therefore, the general threshold model uses the minimum number of parameters
	among the equivalent hypergraph triggering models.
\end{proof}

By a rough counting, a general threshold model $(V, \{f_v\}_{v\in V})$ needs
	$n\cdot (2^{n-1}-1)$ number of parameters
 	to specify $f_v(S)$ for every $S\in 2^V \setminus \{\emptyset \}$ and $v\in V$.
However, if we want to specify a hypergraph triggering model, 
	we need to assign a probability to every subset of hyperedges pointing to every node $v$.
The number of hyperedges pointing to $v$ is $2^{n-1}-1$, and the number of possible subsets
	of these hyperedges is $2^{2^{n-1}-1}$.
Thus roughly we need $n\cdot 2^{2^{n-1}-1}$ parameters to specify a hypergraph triggering model
	(or an equivalent node-independent SBFD model).
Thus, in terms of the number of parameters, general threshold model has exponential savings than
	a fully expressed hypergraph triggering model.
In other words, each general threshold model may have multiple parametric representation
	in the SHD (or SBFD) models that are all equivalent.

\section{SHD Model and Correlated General Threshold Model}

  
In the previous section, we show that the node-independent SHD model, a.k.a.\ hypergraph 
  triggering model, is equivalent to the general threshold model.
In this section, we investigate the relationship between the general SHD model and
  the correlated general threshold model in which node thresholds could be correlated.
The correlated general threshold model is defined below.

\begin{definition}[Correlated General Threshold Model]
  \label{def:CGT}
  A {\em correlated general threshold (CGT) model} 
  is defined as a tuple $(V, \{f_v \}_{v\in V}, \Theta_V)$, where $V$ is a set of $n$ nodes,
  $f_v: 2^{V\setminus \{v\}} \rightarrow [0,1]$ is the threshold function for $v \in V$ 
    with all $f_v$'s being monotone and normalized, 
  and $\Theta_V$ is a joint threshold distribution over $[0,1]^n$.
  The node thresholds $\vtheta = (\theta_v)_{v \in V}$ is sampled from $\Theta_V$ 
  at the beginning, and the rest diffusion process follows the exactly the same
    way as in the general threshold model (Definition~\ref{def:GT}).    
\end{definition}

We remark that making the threshold variables $\theta_v$'s correlated is certainly a straightforward and natural way of generalizing the general threshold model to the correlated activation setting,
	but of course readers may come up with other ways of such a generalization.

One question is whether the CGT model would be equivalent to the SHD model.
By a similar argument as in Lemma~\ref{lem:GTisSBFD}, we could show that the CGT model
  can be represented in the SHD (or SBFD) model.
However, the reverse is not true, as shown by the following example.

\begin{theorem}
    There is an SHD model instance which cannot be represented in the CGT model.
\end{theorem}
\begin{proof}
    We show an even stronger example, a digraph ``reverse triggering'' instance that cannot be represented in the CGT model.
    Consider nodes $V = \{u_1, u_2, v_1, v_2\}$, where there are possible edges going from $u_1$ and $u_2$ to $v_1$ and $v_2$.
    We specify the distributions over outgoing edges from $u_1$ and $u_2$ respectively.
    For $u_1$, with probability $0.5$, there is an edge $(u_1, v_1)$, and otherwise, there is an edge $(u_1, v_2)$.
    For $u_2$, with probability $0.1$, there are edges $(u_2, v_1)$ and $(u_2, v_2)$, and otherwise there is no outgoing edge.

    Now suppose there is an equivalent CGT model instance.
    Consider the threshold functions of $v_1$ and $v_2$.
    Observe that:
    \begin{itemize}
        \item Letting $\{u_1\}$ be the seed set, we have $\Pr\{\theta_{v_1} \le f_{v_1}(\{u_1\})\} = \Pr\{\theta_{v_2} \le f_{v_2}(\{u_1\})\} = 0.5$.
        \item Letting $\{u_2\}$ be the seed set, we have $\Pr\{\theta_{v_1} \le f_{v_1}(\{u_2\})\} = \Pr\{\theta_{v_2} \le f_{v_2}(\{u_2\})\} = 0.1$.
        \item Therefore we know that $f_{v_1}(\{u_1\}) \ge f_{v_1}(\{u_2\})$, and $f_{v_2}(\{u_1\}) \ge f_{v_2}(\{u_2\})$.
        \item Now consider the joint activation of $v_1$ and $v_2$.
        When $\{u_1\}$ is the seed set, $\Pr\{\theta_1 \le f_{v_1}(\{u_1\}) \wedge \theta_2 \le f_{v_2}(\{u_1\})\} = 0$, because $u_1$ never activates $v_1$ and $v_2$ simultaneously.
        When $\{u_2\}$ is the seed set, $\Pr\{\theta_1 \le f_{v_1}(\{u_2\}) \wedge \theta_2 \le f_{v_2}(\{u_2\})\} = 0.1$.
        On the other hand, since $f_{v_1}(\{u_1\}) \ge f_{v_1}(\{u_2\})$ and $f_{v_2}(\{u_1\}) \ge f_{v_2}(\{u_2\})$, we have
        \[
            0 = \Pr\{\theta_1 \le f_{v_1}(\{u_1\}) \wedge \theta_2 \le f_{v_2}(\{u_1\})\} \ge \Pr\{\theta_1 \le f_{v_1}(\{u_2\}) \wedge \theta_2 \le f_{v_2}(\{u_2\})\} = 0.1,
        \]
        a contradiction.
    \end{itemize}
    Thus we conclude that the ``reverse triggering'' instance constructed above cannot be represented in the CGT model.
\end{proof}

Note that in the above proof, we do not even use the hypergraph construction.
Thus, in effect it shows that CGT is not strong enough to represent live-edge graph models where income live-edges of different nodes may be correlated. 
This means that the correlation only on the threshold values of different nodes is still weak in describing more complicated correlations in the general models as the SHD model.

\section{Discussions and Future Work}

In this paper, we show that the general threshold model is equivalent to hypergraph triggering model or 
	Boolean-function triggering model.
This indicates that the general threshold model exactly covers the group influence modeled by the hypergraph model.
Studying the relationship among these models could help us in better understanding the characteristics of the models.

There are a number of further questions we can ask based on the results of this paper. 
For example, the general threshold model becomes submodular when every threshold function is submodular.
What is the submodular representation under the hypergraph triggering model?
The triggering model allows efficient influence maximization algorithms based on the reverse
	influence sampling approach~\cite{BorgsBrautbarChayesLucier,tang14,tang15}.
Can we expect that, by utilizing hypergraph triggering model, we may have some efficient
	algorithm for the general threshold model?
Investigating these questions may help us to further enhance our knowledge on influence propagation models
	and their algorithmic implications.

\bibliographystyle{abbrv}
\bibliography{bibdatabase} 

\end{document}